\documentclass[submission,copyright,creativecommons]{eptcs}
\providecommand{\event}{Non-Classical Logics 2024 (NCL'24)}

\usepackage{iftex}

\ifpdf
  \usepackage{underscore}         % Only needed if you use pdflatex.
  \usepackage[T1]{fontenc}        % Recommended with pdflatex
\else
  \usepackage{breakurl}           % Not needed if you use pdflatex only.
\fi

\usepackage{amsmath,amsthm,amsfonts,amssymb,mathrsfs,mathtools}
\usepackage{bussproofs}
\usepackage{enumitem}
\usepackage{booktabs}
\usepackage{stmaryrd}

\theoremstyle{definition}
\newtheorem{definition}{Definition}
\newtheorem{theorem}{Theorem}

\newtheorem{corollary}{Corollary}
\newtheorem{proposition}{Proposition}
\newtheorem{remark}{Remark}
\newtheorem{example}{Example}

\newcommand{\all}{\forall}                   % 全称量化子
\newcommand{\some}{\exists}                  % 存在量化子
\newcommand{\oto}{\leftrightarrow}           % 一重同値
\newcommand{\N}{\mathbb{N}}                  % 自然数
\newcommand{\setof}[1]{\{\,{#1}\,\}}
\newcommand{\inset}[2]{\left\{\, {#1} \,|\, {#2} \,\right\}}
\newcommand{\tuple}[1]{\langle {#1} \rangle}

\newcommand{\Hi}{\mathsf{H}}                 % Hilbert system
\newcommand{\K}[1][]{\mathbf{K#1}}

\newcommand{\LAN}{\mathcal{L}}
\newcommand{\VAR}{\mathtt{Var}}
\newcommand{\CON}{\mathtt{Cn}}
\newcommand{\FUN}{\mathtt{Fn}}

\newcommand{\REL}{\mathtt{Rel}}

\newcommand{\TYPE}{\mathtt{TYPE}}

\newcommand{\type}{\mathtt{t}}
\newcommand{\agt}{\mathtt{agt}}
\newcommand{\obj}{\mathtt{obj}}
\newcommand{\agtobj}{\mathtt{agt\_or\_obj}}

\newcommand{\ue}{\mathsf{UE}}
\newcommand{\id}{\mathsf{Id}}
\newcommand{\ps}{\mathsf{PS}}
\newcommand{\eid}{\some\mathsf{Id}}
\newcommand{\dd}{\mathsf{DD}}
\newcommand{\axk}{\mathsf{K}}
\newcommand{\barcan}{\mathsf{BF}}
\newcommand{\kni}{\mathsf{KNI}}
\newcommand{\rmp}{\mathsf{MP}}
\newcommand{\rkg}{\mathsf{KG}}
\newcommand{\rug}{\mathsf{UG}}

\title{Semantic Incompleteness of \\
Liberman et al. (2020)'s Hilbert-style System for \\
Term-modal Logic $\mathbf{K}$ with Equality and Non-rigid Terms}
\author{Takahiro Sawasaki
\institute{Institute of Liberal Arts and Science\\
Kanazawa University\\
Kanazawa, Japan}
\email{tsawasaki@staff.kanazawa-u.ac.jp}
}
\def\titlerunning{Semantic incompleteness of Liberman et al. (2020)'s Hilbert-style system}
\def\authorrunning{T. Sawasaki}
\begin{document}
\maketitle

\begin{abstract}
    In this paper, we prove the semantic incompleteness of the Hilbert-style system for the minimal normal term-modal logic with equality and non-rigid terms that was proposed in Liberman~et~al. (2020) ``Dynamic Term-modal Logics for First-order Epistemic Planning.'' Term-modal logic is a family of first-order modal logics having term-modal operators indexed with terms in the first-order language. While some first-order formula is valid over the class of all frames in the Kripke semantics for the term-modal logic proposed there, it is not derivable in Liberman~et~al. (2020)'s Hilbert-style system. We show this fact by introducing a non-standard Kripke semantics which makes the meanings of constants and function symbols relative to the meanings of relation symbols combined with them.
\end{abstract}

\section{Introduction}
\label{sec:intro}

In this paper, we prove the semantic incompleteness of the Hilbert-style system $\Hi\K$ for the minimal normal term-modal logic $\K$ with equality and non-rigid terms that was proposed in Liberman et al.~\cite{Liberman2020}. Term-modal logic, developed by Thalmann~\cite{Thalmann2000} and Fitting et al.~\cite{Fitting2001}, is a family of first-order modal logics having term-modal operators $[t]$ indexed with terms $t$ in the first-order language. In the language of term-modal logic, for example, $[x]P(x)$, $[f(x)]P(x)$ and $\all x [f(x)]P(x)$ are formulas. Term-modal logic is more expressive than multi-modal propositional logic and has been applied to epistemic logic \cite{Kooi2008,Rendsvig2010,Rendsvig2011,Sedlar2014,Achen2017,Wang2018,Wang2022,Naumov2019,Naumov2023,Liberman2019,Liberman2020,Liberman2020a} and deontic logic \cite{Sawasaki2019a,Sawasaki2021,Frijters2019,Frijters2020,Frijters2021a,Frijters2021,Frijters2023}. Some other developments of term-modal logic have been overviewed e.g. in \cite[pp. 22-4]{Liberman2020} and \cite[pp. 48-50]{Frijters2021}.

The logic developed in Liberman et al.~\cite{Liberman2020} is a first-order dynamic epistemic logic for epistemic planning, and term-modal logic is invoked as its underlying logic. Technically speaking, their term-modal logic is a two-sorted normal term-modal logic of the constant domain with equality and non-rigid terms. They make their logic two-sorted because, whereas letting the domain of a model include both agents and objects, they read an epistemically interpreted term-modal operator $K_t$ as ``\emph{agent} $t$ knows.'' The language defined in \cite{Liberman2020} allows $K_t \varphi$ to be a formula only if $t$ is a term for an agent, and thereby excludes the possibility that terms denoting objects appear in the argument of the term-modal operator. The Hilbert-style system $\Hi\K$ found in \cite[p. 17]{Liberman2020} was originally presented in \cite{Achen2017} which is probably based on \cite{Rendsvig2010,Rendsvig2011}. Later, two issues on action model and reduction axiom were fixed in the erratum \cite{Liberman2023erratum} of \cite{Liberman2020}.

Unfortunately, $\Hi\K$ is semantically incomplete due to the unprovability of a first-order formula $x = c \to (P(x) \to P(c))$. In Section~\ref{sec:semantic-incomp}, we show that it is valid over the class of all frames whereas it is unprovable in $\Hi\K$. To this end, we there introduce a non-standard Kripke semantics which makes the meanings of constants and function symbols relative to the meanings of relation symbols combined with them.

It is worth noting here that, as the above first-order formula suggests, the semantic incompleteness of $\Hi\K$ is irrelevant to its term-modal aspects. To make the point clear, let $\LAN$ be a first-order modal language having equality, constants and only the ordinary non-indexed modal operator as its modal operators, and say the semantics for first-order modal logic (the FOML-semantics for short) to refer to the Kripke semantics of the constant domain given to $\LAN$ in which the accessibility relation is just a binary relation on worlds and constants are interpreted relative to worlds. Using a semantics similar to the non-standard semantics introduced in Section~\ref{sec:semantic-incomp}, we can in fact prove that the Hilbert-style system naturally obtained from $\Hi\K$ by changing from the two-sorted term-modal language to $\LAN$ becomes semantically incomplete with respect to the FOML-semantics similarly due to the unprovability of $x = c \to (P(x) \to P(c))$. The question to be asked here is what is the formulation of a sound and complete Hilbert-style system to the FOML-semantics. To the best of our knowledge, this is still an open question. Such a Hilbert-style system seems to have never been provided together with a detailed proof in the literature.\footnote{As a sound and complete proof system with respect to the multi-modal FOML-semantics with the epistemic accessibility relation for each agent, Fagin~et~al.~\cite[p. 90]{Fagin2003} offered a Hilbert-style system having two first-order principles $A(t/x) \to \some x A$ and $t = s \to (A(t/z) \oto A(s/z))$ as axioms with a restriction that $t,s$ must be variables if $A$ has any occurrence of an (not term-modal) epistemic operator $K_{a}$. However, the proof of this system's completeness is omitted there.}

This paper will proceed as follows.
In Section~\ref{sec:liberman2020} we first introduce the syntax in \cite{Liberman2020}. Since there are some minor defects on the definitions for type, we do this with some modifications. Then we introduce the Kripke semantics and the Hilbert-style system $\Hi\K$ given in \cite{Liberman2020}.
In Section~\ref{sec:semantic-incomp} we prove the semantic incompleteness of $\Hi\K$ by introducing a non-standard Kripke semantics for which $\Hi\K$ is sound but in which $x = c \to (P(x) \to P(c))$ is not valid.

\section{Syntax, Semantics and the Hilbert-style System $\Hi\K$}
\label{sec:liberman2020}

We will first introduce the syntax presented in \cite[pp. 3-4]{Liberman2020} with some modifications. The idea there is to define the notions of term and formula while assigning (sequences of) types ``$\agt$'', ``$\obj$'' or ``$\agtobj$'' to all symbols like variables or relation symbols. It is basically the same idea as in Enderton~\cite[Section~4.3]{Enderton2001}, but there is an important difference. In the syntax of \cite{Liberman2020}, not only $\agt$ or $\obj$ but also $\agtobj$ may be assigned to the arguments of function symbols and relation symbols, so that $P(x)$ seems to be intended to become a formula even when $x$ has type $\agt$ and $P$ takes type $\agtobj$.

However, the original definitions~1--3 for the syntax seem to have two minor defects. First, the original definition~1 for type assignment and the original definition~2 for term are dependent upon one another, thus they are circular definitions. Second, whereas $P(x)$ seems to be intended to become a formula when $x$ has type $\agt$ and $P$ takes type $\agtobj$, it does not actually become a formula since the original definition~3 for formula requires that the type of $x$ and the type of the argument of $P$ must be the same. Accordingly, for example, $x = x$ cannot be a formula in any signature since the type of $x$ is either $\agt$ or $\obj$ but the type of the arguments of $=$ is always $\agtobj$.

To amend the above two defects, we redefine the syntax in \cite[pp. 3-4]{Liberman2020} as follows.

\begin{definition}[Signature]
\label{def:sig-type}
    Let $\VAR$ be a countably infinite set of \emph{variables}, $\CON$ a countable set of \emph{constants}, $\FUN$ a countable set of \emph{function symbols}, and $\REL$ a countable set of \emph{relation symbols} containing the \emph{equality symbol} $=$. Let $\tuple{\TYPE,\preccurlyeq}$ be also the ordered set of \emph{types} where $\TYPE$ $=$ $\setof{\agt,\obj,\agtobj}$ and $\preccurlyeq$ is the reflexive ordering on $\TYPE$ with $\agt \preccurlyeq \agtobj$ and $\obj \preccurlyeq \agtobj$, i.e.,
    \[
    \preccurlyeq \,\,\coloneqq\,\, \inset{\tuple{\tau,\tau}}{\tau \in \TYPE} \cup \setof{\tuple{\agt,\agtobj},\tuple{\obj,\agtobj}}.
    \]
    A \emph{type assignment} $\type \colon \VAR \cup \CON \cup \FUN \cup \REL$ $\to$ $\bigcup_{n \in \N} \TYPE^{n}$ is an assignment mapping
    \begin{enumerate}
    \renewcommand{\labelenumi}{\arabic{enumi}.}
        \item a variable $x$ to a type $\type(x) \in \setof{\agt,\obj}$ such that both $\VAR \cap \type^{-1}[\setof{\agt}]$ and $\VAR \cap \type^{-1}[\setof{\obj}]$ are countably infinite, where $\type^{-1}[X]$ is the inverse image of a set $X$;
        \item a constant $c$ to a type $\type(c) \in \setof{\agt,\obj}$;
        \item a function symbol $f$ to a sequence of types $\type(f) \in \TYPE^{n} \times \setof{\agt,\obj}$ for some $n \in \N$;
        \item the equality symbol $=$ to the sequence of types $\type(=) = \tuple{\agtobj,\agtobj}$;
        \item a relation symbol $P$ distinct from $=$ to a sequence of types $\type(P) \in \TYPE^n$ for some $n\in\N$.
    \end{enumerate}
    The tuple $\tuple{\VAR,\CON,\FUN,\REL,\type}$ is called a \emph{signature}.
\end{definition}

\begin{definition}[Term of Type]
\label{def:term}
    Let $\tuple{\VAR,\CON,\FUN,\REL,\type}$ be a signature. The set of \emph{terms of types} is defined as follows.
    \begin{enumerate}
    \renewcommand{\labelenumi}{\arabic{enumi}.}
        \item any variable $x \in \VAR$ is a term of type $\type(x)$.
        \item any constant $c \in \CON$ is a term of type $\type(c)$.
        \item If $t_1,\dots,t_n$ are terms of types $\tau_1,\dots,\tau_n$ and $f$ is a function symbol in $\FUN$ such that $\type(f)$ $=$ $\tuple{\tau'_{1},\dots,\tau'_{n},\tau'_{n+1}}$ and $\tau_i \preccurlyeq \tau'_i$, then $f(t_1,\dots,t_n)$ is a term of type $\tau'_{n+1}$.
    \end{enumerate}
    For convenience, henceforth we use a type assignment $\type$ to mean its uniquely extended assignment by letting $\type(f(t_{1},\dots,t_{n}))$ $=$ $\tau$ for each term of the form $f(t_{1},\dots,t_{n})$ of type $\tau$.
\end{definition}

\begin{definition}[Language]
\label{def:language}
    Let $\tuple{\VAR,\CON,\FUN,\REL,\type}$ be a signature. The \emph{language} is the set of formulas $\varphi$ defined in the following BNF.
    \[
        \varphi \Coloneqq P(t_1,\dots,t_n) \mid \neg \varphi \mid \varphi \land \varphi \mid K_{s}\varphi \mid \all x \varphi,
    \]
    where $t_1,\dots,t_n,s$ are terms with $\type(s)$ $=$ $\agt$ and $P \in \REL$ such that $\type(P)$ $=$ $\tuple{\tau_1,\dots,\tau_n}$ and $\type(t_i) \preccurlyeq \tau_i$. Note here that $P$ can be $=$.
\end{definition}

\noindent
As usual, we use the notations $t \neq s$ $\coloneqq \neg (t=s)$, $\varphi \to \psi$ $\coloneqq$ $\neg(\varphi \land \neg \psi)$ and $\some x \varphi$ $\coloneqq$ $\neg \all x \neg \varphi$.

We believe that our definitions successfully capture what was intended in the original definitions~1--3. On top of these definitions, we will follow \cite[p. 4]{Liberman2020} to define the notions of \emph{free variable} and \emph{bound variable} in a formula as usual, where the set of free variables in $K_{t}\varphi$ is defined as the union of the set of variables in $t$ and the set of free variables in $\varphi$. For a variable $x$, terms $t,s$ and a formula $\varphi$ such that $\type(x) = \type(s)$ and no variables in $s$ are bound variables in $\varphi$, we also define \emph{substitutions} $t(s/x)$ and $\varphi(s/x)$ of $s$ for $x$ in $t$ and $\varphi$ in a usual manner, except that $(K_{t}\varphi)(s/x)$ $=$ $K_{t(s/x)}\varphi(s/x)$. Whenever we write $t(s/x)$ or $\varphi(s/x)$, we tacitly assume that $\type(x) = \type(s)$ and no variables in $s$ are bound variables in $\varphi$. We also define the lengths of term and formula as usual.

Let us now introduce the Kripke semantics presented in \cite[pp. 5-6]{Liberman2020}.

\begin{definition}[Frame, {\cite[Def. 4]{Liberman2020}}]
\label{def:frame}
    A \emph{frame} is a tuple $F = \tuple{D,W,R}$ where
    \begin{enumerate}
    \renewcommand{\labelenumi}{\arabic{enumi}.}
        \item $D \coloneqq D_{\agtobj} \coloneqq D_{\agt} \sqcup D_{\obj}$ is the disjoint union of a non-empty set $D_{\agt}$ of \emph{agents} and a non-empty set $D_{\obj}$ of \emph{objects};
        \item $W$ is a non-empty set of \emph{worlds};
        \item $R$ is a mapping that assigns to each agent $i \in D_{\agt}$ a binary relation $R_{i}$ on $W$, i.e., $R \colon D_{\agt} \to \mathcal{P}(W \times W)$.
    \end{enumerate}
\end{definition}

\begin{definition}[Model, {\cite[Def. 5]{Liberman2020}}]
\label{def:model}
    Let $\tuple{\VAR,\CON,\FUN,\REL,\type}$ be a signature. A \emph{model} is a tuple $M = \tuple{D,W,R,I}$ where $\tuple{D,W,R}$ is a frame and $I$ is an \emph{interpretation} that maps
    \begin{enumerate}
    \renewcommand{\labelenumi}{\arabic{enumi}.}
        \item a pair $\tuple{c,w}$ of some $c \in \CON$ and some $w \in W$ to an element $I(c,w) \in D_{\type(c)}$;
        \item a pair $\tuple{f,w}$ of some $f \in \FUN$ and some $w \in W$ to a function $I(f,w) \colon (D_{\tau_{1}} \times \cdots \times D_{\tau_{n}}) \to D_{\tau_{n+1}}$, where $\type(f)$ $=$ $\tuple{\tau_{1},\dots,\tau_{n},\tau_{n+1}}$;
        \item a pair $\tuple{=,w}$ of the equality symbol $=$ and some $w \in W$ to the set $I(=,w)$ $=$\linebreak $\inset{\tuple{d,d}}{d \in D_{\agtobj}}$;
        \item a pair $\tuple{P,w}$ of some $P \in \REL \setminus \setof{=}$ and some $w \in W$ to a subset $I(P,w)$ of $D_{\tau_{1}} \times \cdots \times D_{\tau_{n}}$, where $\type(P)$ $=$ $\tuple{\tau_1,\dots,\tau_n}$.
    \end{enumerate}
\end{definition}

\begin{definition}[Valuation, {\cite[Def. 6, 7]{Liberman2020}}]
\label{def:valuation}
    A \emph{valuation} is a mapping $v \colon \VAR \to D$ such that $v(x) \in D_{\type(x)}$ and the valuation $v[x \mapsto d]$ is the same valuation as $v$ except for assigning to a variable $x$ an element $d \in D_{\type(x)}$. Given a valuation $v$, a world $w$ and an interpretation $I$ in a model, the extension $\llbracket t \rrbracket^{I,v}_{w}$ of a term $t$ is defined by $\llbracket x \rrbracket^{I,v}_{w}$ $=$ $v(x)$, $\llbracket c \rrbracket^{I,v}_{w}$ $=$ $I(c,w)$, and $\llbracket f(t_{1},\dots,f_n) \rrbracket^{I,v}_{w}$ $=$ $I(f,w)(\llbracket t_{1} \rrbracket^{I,v}_{w},\dots,\llbracket t_{n} \rrbracket^{I,v}_{w})$.
\end{definition}

\begin{definition}[Satisfaction, {\cite[Def. 8]{Liberman2020}}]
\label{def:satisfaction}
    The \emph{satisfaction} $M,w \models_v \varphi$ of a formula $\varphi$ at a world $w$ in a model $M$ under a valuation $v$ is defined as follows.
    \begin{align*}
        &M,w \models_v P(t_1,\dots,t_n) &&\text{iff} &&\tuple{\llbracket t_1 \rrbracket^{I,v}_{w},\dots,\llbracket t_n \rrbracket^{I,v}_{w}} \in I(P,w) \qquad \text{($P$ can be $=$)} \\
        &M,w \models_v \neg \varphi &&\text{iff} &&M,w \not\models_v \varphi \\
        &M,w \models_v \varphi \land \psi &&\text{iff} &&M,w \models_v \varphi \quad \text{and} \quad M,w \models_v \psi \\
        &M,w \models_v \all x \varphi &&\text{iff} &&M,w \models_{v[x \mapsto d]} \varphi \quad \text{for all $d \in D_{\type(x)}$} \\
        &M,w \models_v K_{t} \varphi &&\text{iff} &&M,w' \models_{v} \varphi \quad \text{for all $w' \in W$ such that $\tuple{w,w'} \in R_{\llbracket t \rrbracket^{I,v}_{w}}$}
    \end{align*}
\end{definition}

\begin{definition}[Validity, {\cite[p. 25]{Liberman2020}}]
\label{def:validity}
    A formula $\varphi$ is \emph{valid} if for all models $M$, all worlds $w \in W$ and all valuations $v$, it holds that $M,w \models_{v} \varphi$.
\end{definition}

\begin{remark}
    Instead of the $x$-variant of a valuation $v$ used in \cite{Liberman2020}, we adopted the valuation $v[x \mapsto d]$ to give the satisfaction for $\all x \varphi$. This change is just for the clarity of our proof and does not affect the satisfiability of formulas. As for validity, because unlike \cite{Liberman2020} we are only interested here in the validity of formula over the class of all frames, for the sake of brevity we defined the validity of formula independently of any class of frames.
\end{remark}

\noindent
For ease of reference, henceforth we call this semantics \emph{TML-semantics}.

Finally, we will introduce by Table~\ref{table:logicK} the Hilbert-style system $\Hi\K$ for the minimal normal term-modal logic $\K$ presented in Liberman et al.~\cite[p. 17]{Liberman2020}. The notion of provability is defined as usual.

\begin{table}[htb]
% \renewcommand{\arraystretch}{1.5}
% \multicolumn{跨るセル数}{寄せ}{セルの中身}
% @{\hspace{1cm}}でセル間に1cmの余白
\begin{tabular}{ll ll}
    \toprule
    \multicolumn{4}{l}{\textbf{Axiom}} \\

    \mbox{} & all propositional tautologies & \mbox{} & \mbox{} \\
    $\ue$ & $\all x \varphi \to \varphi(y/x)$ &
        $\axk$ & $K_{t}(\varphi \to \psi) \to (K_{t} \varphi \to K_{t} \psi)$ \\
    $\id$ & $t=t$ &
        $\barcan$ & $\all x K_{t} \varphi \to K_{t}\all x \varphi$ \quad for $x$ not occurring in $t$ \\
    $\ps$ & $x=y \to (\varphi(x/z) \to \varphi(y/z))$ &
        $\kni$ & $x \neq y \to K_{t} x \neq y$ \\
    $\eid$ & $c=c \to \some x (x=c)$ & \mbox{} & \mbox{} \\
    $\dd$ & $x \neq y$ \quad if $\type(x) \neq \type(y)$ & \mbox{} & \mbox{} \\

    \multicolumn{4}{c}{\mbox{}} \\

    \multicolumn{4}{l}{\textbf{Inference rules}} \\
    $\rmp$ & \multicolumn{3}{l}{From $\varphi$ and $\varphi \to \psi$, infer $\psi$} \\
    $\rkg$ & \multicolumn{3}{l}{From $\varphi$, infer $K_{t} \varphi$} \\
    $\rug$ & \multicolumn{3}{l}{From $\varphi \to \psi$, infer $\varphi \to \all x \psi$ \quad for $x$ not free in $\varphi$} \\
    \bottomrule
\end{tabular}
\centering
\caption{The Hilbert-style system $\Hi\K$ for the minimal normal term-modal logic $\K$}
\label{table:logicK}
\end{table}

What is involving the semantic incompleteness of $\Hi\K$ here is $\ue$ and $\ps$. As remarked in Fagin et al.~\cite[pp. 88-9]{Fagin2003}, the ordinary first-order axioms $\all x \varphi \to \varphi(t/x)$ and $t = s \to (\varphi(t/z) \to \varphi(s/z))$ are not valid in Kripke semantics for first-order modal logic where constants or function symbols are interpreted as non-rigid.
In order to avoid making invalid formulas provable, Liberman et al.~\cite{Liberman2020} adopted the variable-restricted versions $\ue$ and $\ps$ of these two axioms.
The problem is that $\ps$ or its combinations with $\ue$ or $\eid$ are not sufficient to derive a valid formula $x = c \to (P(x) \to P(c))$.

\section{Semantic Incompleteness of the Hilbert-style System $\Hi\K$}
\label{sec:semantic-incomp}

In this section, we prove the semantic incompleteness of $\Hi\K$ by showing that $x = c \to (P(x) \to P(c))$ is valid in the TML-semantics but not provable in $\Hi\K$. As expected, there is no difficulty to show the former.

\begin{proposition}
\label{prop:validity-of-unprovable-formula}
    Let $\tuple{\VAR,\CON,\FUN,\REL,\type}$ be a signature, $x \in \VAR$, $c \in \CON$ and $P \in \REL$ with $\type(P)$ $=$ $\tuple{\agtobj}$.
    A formula $x = c \to (P(x) \to P(c))$ is valid in the TML-semantics.
\end{proposition}

\begin{proof}
    Suppose $M,w \models_v x = c$ and $M,w \models_v P(x)$. Since $\llbracket x \rrbracket^{I,v}_{w} = \llbracket c \rrbracket^{I,v}_{w}$ and $\llbracket x \rrbracket^{I,v}_{w} \in I(P,w)$, we have $\llbracket c \rrbracket^{I,v}_{w} \in I(P,w)$. Thus $M,w \models_v P(c)$.
\end{proof}

To establish the unprovability of $x = c \to (P(x) \to P(c))$, it is sufficient to find a new semantics to which $\Hi\K$ is sound but in which this formula is not valid. To this end, we will first introduce the notion of non-standard model as follows.

\begin{definition}[Non-standard Model]
\label{def:col-model}
    Let $\tuple{\VAR,\CON,\FUN,\REL,\type}$ be a signature. A \emph{non-standard model} is a tuple $N = \tuple{D,W,R,J}$ where $\tuple{D,W,R}$ is a frame in the sense of Definition~\ref{def:frame} and $J$ is an interpretation that maps
    \begin{enumerate}
    \renewcommand{\labelenumi}{\arabic{enumi}.}
        \item a triple $\tuple{c,w,X}$ of some $c \in \CON$, some $w \in W$ and some $X \subseteq D^{n}$ for some $n \in \N$ to an element $J(c,w,X) \in D_{\type(c)}$;
        \item a triple $\tuple{f,w,X}$ of some $f \in \FUN$, some $w \in W$ and some $X \subseteq D^{n}$ for some $n \in \N$ to a function $J(f,w,X) \colon (D_{\tau_{1}} \times \cdots \times D_{\tau_{n}}) \to D_{\tau_{n+1}}$, where $\type(f)$ $=$ $\tuple{\tau_{1},\dots,\tau_{n+1}}$;
        \item a pair $\tuple{=,w}$ of the equality symbol $=$ and some $w \in W$ to the set $J(=,w)$ $=$\linebreak $\inset{\tuple{d,d}}{d \in D_{\agtobj}}$;
        \item a pair $\tuple{P,w}$ of some $P \in \REL \setminus \setof{=}$ and some $w \in W$ to a subset $J(P,w)$ of $D_{\tau_{1}} \times \cdots \times D_{\tau_{n}}$, where $\type(P)$ $=$ $\tuple{\tau_1,\dots,\tau_n}$.
    \end{enumerate}
\end{definition}

\noindent
Here is the intuition. A subset $X$ of $D^n$ is a set of sequences consisting of either/both of agents and objects. Thus, the set $X$ mentioned in the meanings $J(c,w,X)$ and $J(f,w,X)$ of a constant $c$ and a function symbol $f$ can serve as the meaning of a relation symbol. This trick enables us to make the meanings of constants and function symbols relative to the meanings of relation symbols combined with them.

We then define the notion of satisfaction of formula in non-standard model. In what follows, we use the same notion of valuation as in the TML-semantics and define the extension $\llbracket t \rrbracket^{J,v}_{w,X}$ of a term $t$ in a given non-standard model similarly by letting $\llbracket x \rrbracket^{J,v}_{w,X}$ $=$ $v(x)$, $\llbracket c \rrbracket^{J,v}_{w,X}$ $=$ $J(c,w,X)$ and $\llbracket f(t_1,\dots,t_n) \rrbracket^{J,v}_{w,X}$ $=$ $J(f,w,X)(\llbracket t_1 \rrbracket^{J,v}_{w,X},\dots,\llbracket t_n \rrbracket^{J,v}_{w,X})$.

\begin{definition}[Satisfaction in Non-standard Model]
\label{def:col-satisfaction}
    The \emph{satisfaction} $N,w \models_v \varphi$ of a formula $\varphi$ at a world $w$ in a non-standard model $N$ under a valuation $v$ is defined as follows.
    \begin{align*}
        &N,w \models_v P(t_1,\dots,t_n) &&\text{iff} &&\tuple{\llbracket t_1 \rrbracket^{J,v}_{w,J(P,w)},\dots,\llbracket t_n \rrbracket^{J,v}_{w,J(P,w)}} \in J(P,w) \qquad \text{($P$ can be $=$)} \\
        &N,w \models_v \neg \varphi &&\text{iff} &&N,w \not\models_v \varphi \\
        &N,w \models_v \varphi \land \psi &&\text{iff} &&N,w \models_v \varphi \quad \text{and} \quad N,w \models_v \psi \\
        &N,w \models_v \all x \varphi &&\text{iff} &&N,w \models_{v[x \mapsto d]} \varphi \quad \text{for all $d \in D_{\type(x)}$} \\
        &N,w \models_v K_{t} \varphi &&\text{iff} &&N,w' \models_{v} \varphi \quad \text{for all $w' \in W$ such that $\tuple{w,w'} \in R_{\llbracket t \rrbracket^{J,v}_{w,\emptyset}}$}
    \end{align*}
\end{definition}

\noindent
What we should pay attention here is the satisfactions of atomic formula $P(t_{1},\dots,t_{n})$ and term-modal formula $K_{t}\varphi$. In the satisfaction of $P(t_{1},\dots,t_{n})$ in non-standard model, the meaning $\llbracket t_{i} \rrbracket^{J,v}_{w,J(P,w)}$ of each $t_{i}$ in $P(t_{1},\dots,t_{n})$ is determined by the interpretation $J$, the valuation $v$, the world $w$ and \emph{the meaning $J(P,w)$ of the relation symbol $P$ combined with terms $t_{1},\dots,t_{n}$}. Thus, as explained in the following Example~\ref{ex:lewis}, the meaning of a constant $c$ occurring in $P(c)$ could be different from that of $c$ occurring in $Q(c)$.

\begin{example}
\label{ex:lewis}
    Let $lewis \in \CON$ with $\type(lewis) = \agt$ and $SL,CF \in \REL$ with $\type(SL) = \type(CF) = \tuple{\agt}$, and consider a non-standard model such that
    \begin{align*}
        J(SL,w) &= \{i \in D_{\agt} \,|\, \text{$i$ is one of the authors of \textit{Symbolic Logic}} \}, \\
        J(CF,w) &= \inset{i \in D_{\agt}}{\text{$i$ is the author of \textit{Counterfactuals}}},
    \end{align*}
    $J(lewis,w,J(SL,w))$ is C. I. Lewis and $J(lewis,w,J(CF,w))$ is D. Lewis. The meaning\linebreak $J(lewis,w,J(SL,w))$ of $lewis$ occurring in $SL(lewis)$ is then different from the meaning\linebreak $J(lewis,w,J(CF,w))$ of $lewis$ occurring in $CF(lewis)$. Note that, although $J(lewis,w,J(SL,w)) \in J(SL,w)$ holds in the above non-standard model, we can technically have a non-standard model such that\linebreak $J(lewis,w,J(SL,w)) \notin J(SL,w)$ holds by assigning D. Lewis to $J(lewis,w,J(SL,w))$.
\end{example}

\noindent On the other hand, because the meaning $\llbracket t \rrbracket^{J,v}_{w,\emptyset}$ of $t$ in $K_t$ is determined independently of the meaning of any relation symbol, the satisfaction of $K_{t}\varphi$ in non-standard model is in effect the same as the satisfaction of $K_{t}\varphi$ in model of the TML-semantics. By this fact we can validate axioms $\axk$ and $\barcan$ in this semantics.

The notion of validity is defined as in the TML-semantics. For ease of reference, henceforth we call this semantics \emph{non-standard semantics}.

Now it is easy to see the invalidity of $x = c \to (P(x) \to P(c))$ in the non-standard semantics.

\begin{proposition}
\label{prop:invalidity-of-unprovable-formula}
    Let $\tuple{\VAR,\CON,\FUN,\REL,\type}$ be a signature, $x \in \VAR$, $c \in \CON$ with $\type(x) = \type(c)$ and $P \in \REL$ with $\type(P)$ $=$ $\tuple{\agtobj}$.
    A formula $x = c \to (P(x) \to P(c))$ is not valid in the non-standard semantics.
\end{proposition}

\begin{proof}
    We may assume $\type(x) = \type(c) = \agt$ without loss of generality.
    Let $N = \tuple{D,W,R,J}$ be a non-standard model such that $w \in W$, $D_{\agt}$ $=$ $\setof{\alpha,\beta}$, $J(c,w,\{\,\tuple{d,d}\,\mid$ $d \in D_{\agtobj}\,\})$ $=$ $\alpha$, $J(c,w,\setof{\alpha})$ $=$ $\beta$ and $J(P,w)$ $=$ $\setof{\alpha}$. Let $v$ be also a valuation such that $v(x) = \alpha$. Since
    \[
        \llbracket x \rrbracket^{J,v}_{w,J(=,w)} = v(x) = \alpha = J(c,w,\{\,\tuple{d,d}\,\mid d \in D_{\agtobj}\,\}) = J(c,w,J(=,w)) = \llbracket c \rrbracket^{J,v}_{w,J(=,w)},
    \]
    we have $N,w \models_v x = c$. It is also easy to see $N,w \models_v P(x)$. However, since
    \[
        \llbracket c \rrbracket^{J,v}_{w,J(P,w)} = J(c,w,J(P,w)) = J(c,w,\setof{\alpha}) = \beta,
    \]
    it fails that $N,w \models_v P(c)$. Therefore $x=c \to (P(x) \to P(c))$ is not valid in the non-standard semantics.
\end{proof}

On top of this, we can prove as below that $\Hi\K$ is sound with respect to the non-standard semantics.

\begin{proposition}
\label{prop:valuation}
    Let $\tuple{\VAR,\CON,\FUN,\REL,\type}$ be a signature and $x,y \in \VAR$ with $\type(x) = \type(y)$. Let $N$ $=$ $\tuple{D,W,R,J}$ be also a non-standard model, $w$ a world, $X$ a subset of $D^n$ for some $n \in \N$ and $v$ a valuation. For all terms $t$,
    \[
        \llbracket t(y/x) \rrbracket^{J,v}_{w,X} \quad = \quad \llbracket t \rrbracket^{J,v[x \mapsto v(y)]}_{w,X}.
    \]
\end{proposition}

\begin{proof}
    By induction on the length of terms.
    \begin{itemize}
        \item For $t$ being the variable $x$, \,\, $\llbracket x(y/x) \rrbracket^{J,v}_{w,X}$ $=$ $v(y)$ $=$ $v[x \mapsto v(y)](x)$ $=$ $\llbracket x \rrbracket^{J,v[x \mapsto v(y)]}_{w,X}$.
        \item For $t$ being a variable $z$ distinct from $x$, \,\, $\llbracket z(y/x) \rrbracket^{J,v}_{w,X}$ $=$ $v(z)$ $=$ $v[x \mapsto v(y)](z)$ $=$ $\llbracket z \rrbracket^{J,v[x \mapsto v(y)]}_{w,X}$.
        \item For $t$ being a constant $c$, \,\, $\llbracket c(y/x) \rrbracket^{J,v}_{w,X}$ $=$ $J(c,w,X)$ $=$ $\llbracket c \rrbracket^{J,v[x \mapsto v(y)]}_{w,X}$.
        \item For $t$ being of the form $f(t_{1},\dots,t_{n})$,
        \begin{align*}
            \llbracket f(t_{1},\dots,t_{n})(y/x) \rrbracket^{J,v}_{w,X} &= \llbracket f(t_{1}(y/x),\dots,t_{n}(y/x)) \rrbracket^{J,v}_{w,X} \\
            &= J(f,w,X)(\llbracket t_{1}(y/x) \rrbracket^{J,v}_{w,X},\dots,\llbracket t_{n}(y/x) \rrbracket^{J,v}_{w,X}) \\
            &= J(f,w,X)(\llbracket t_{1} \rrbracket^{J,v[x \mapsto v(y)]}_{w,X},\dots,\llbracket t_{n} \rrbracket^{J,v[x \mapsto v(y)]}_{w,X}) \quad \text{(inductive hypothesis)} \\
            &= \llbracket f(t_{1},\dots,t_{n}) \rrbracket^{J,v[x \mapsto v(y)]}_{w,X}.
        \end{align*}
    \end{itemize}
\end{proof}

\begin{proposition}
\label{prop:forall}
Let $\tuple{\VAR,\CON,\FUN,\REL,\type}$ be a signature, $x,y \in \VAR$ with $\type(x) = \type(y)$ and $N$ $=$ $\tuple{D,W,R,J}$ a non-standard model. For all worlds $w$, all valuations $v$ and all formulas $\varphi$,
\[
    N,w \models_{v} \varphi(y/x) \quad \text{iff} \quad N,w \models_{v[x \mapsto v(y)]} \varphi.
\]
\end{proposition}

\begin{proof}
    By induction on the length of formulas. Since the proof of the cases for $\neg \psi$ and $\psi \land \gamma$ are straightforward, we see only the cases for $P(t_{1},\dots,t_{n})$, $\all z \psi$ and $K_{t} \psi$.

    \begin{itemize}
        \item For $\varphi$ being of the form $P(t_{1},\dots,t_{n})$,
        \begin{align*}
            N,w \models_{v} P(t_{1},\dots,t_{n})(y/x) \quad &\text{iff} \quad
            \tuple{\llbracket t_{1}(y/x) \rrbracket^{J,v}_{w,J(P,w)},\dots,\llbracket t_{n}(y/x) \rrbracket^{J,v}_{w,J(P,w)}} \in J(P,w) \\
            &\text{iff} \quad \tuple{\llbracket t_{1} \rrbracket^{J,v[x \mapsto v(y)]}_{w,J(P,w)},\dots,\llbracket t_{n} \rrbracket^{J,v[x \mapsto v(y)]}_{w,J(P,w)}} \in J(P,w) \quad \text{(Proposition~\ref{prop:valuation})} \\
            &\text{iff} \quad N,w \models_{v[x \mapsto v(y)]} P(t_{1},\dots,t_{n})
        \end{align*}

        \item For $\varphi$ being of the form $\all z \psi$, if $z = x$, then $N,w \models_{v} (\all x \psi)(y/x)$ iff $N,w \models_{v} \all x \psi$ iff $N,w \models_{v[x \mapsto v(y)]} \all x \psi$. So suppose $z \neq x$. Then
        \begin{align*}
            N,w \models_{v} (\all z \psi)(y/x) \quad &\text{iff} \quad N,w \models_{v} \all z \psi(y/x) \\
            &\text{iff} \quad N,w \models_{v[z \mapsto d]} \psi(y/x) \quad \text{for all $d \in D_{\type(z)}$} \\
            &\text{iff} \quad N,w \models_{v[z \mapsto d][x \mapsto v[z \mapsto d](y)]} \psi \quad \text{for all $d \in D_{\type(z)}$} \quad \text{(inductive hypothesis)} \\
            &\text{iff} \quad N,w \models_{v[x \mapsto v(y)][z \mapsto d]} \psi \quad \text{for all $d \in D_{\type(z)}$} \quad \text{($z \neq x$ and $z \neq y$)} \\
            &\text{iff} \quad N,w \models_{v[x \mapsto v(y)]} \all z \psi \\
        \end{align*}

        \item For $\varphi$ being of the form $K_{t} \psi$,
        \begin{align*}
            N,w \models_{v} (K_{t}\psi)(y/x) \quad &\text{iff} \quad N,w \models_{v} K_{t(y/x)} \psi(y/x) \\
            &\text{iff} \quad N,w' \models_{v} \psi(y/x) \quad \text{all $w' \in W$ such that $\tuple{w,w'} \in R_{\llbracket t(y/x) \rrbracket^{J,v}_{w,\emptyset}}$} \\
            &\text{iff} \quad N,w' \models_{v} \psi(y/x) \quad \text{all $w' \in W$ such that $\tuple{w,w'} \in R_{\llbracket t \rrbracket^{J,v[x \mapsto v(y)]}_{w,\emptyset}}$} \\
            &\phantom{\text{iff}} \quad \text{(Proposition~\ref{prop:valuation})} \\
            &\text{iff} \quad N,w' \models_{v[x \mapsto v(y)]} \psi \quad \text{all $w' \in W$ such that $\tuple{w,w'} \in R_{\llbracket t \rrbracket^{J,v[x \mapsto v(y)]}_{w,\emptyset}}$} \\
            &\phantom{\text{iff}} \quad \text{(inductive hypothesis)} \\
            &\text{iff} \quad N,w \models_{v[x \mapsto v(y)]} K_{t} \psi.
        \end{align*}
    \end{itemize}
\end{proof}

\begin{theorem}[Soundness]
\label{thm:soundness}
    If $\varphi$ is provable in $\Hi\K$, then $\varphi$ is valid in the non-standard semantics.
\end{theorem}

\begin{proof}
    It is sufficient to prove that all axioms are valid and that all inference rules preserve validity. Since the proof of the latter is done as usual, we see only the former.
    \begin{itemize}
        \item For any propositional tautology, its validity is obvious since the non-standard semantics gives the ordinary satisfactions for $\neg$ and $\land$.
        \item For $\ue$, i.e., $\all x \varphi \to \varphi(y/x)$, suppose $N,w \models_v \all x \varphi$. Then $N,w \models_{v[x \mapsto v(y)]} \varphi$. Thus by Proposition~\ref{prop:forall} $N,w \models_v \varphi(y/x)$ holds, as required.
        \item For $\id$, i.e., $t=t$, its validity is obvious.
        \item For $\ps$, i.e., $x=y \to (\varphi(x/z) \to \varphi(y/z))$, its validity is shown by induction on $\varphi$.
        \begin{itemize}
            \item For $\varphi$ being of the form $P(t_1,\dots,t_n)$, suppose $N,w \models_v x = y$ and $N,w \models_v P(t_1,\dots,t_n)(x/z)$. Since
            \[
                \tuple{\llbracket t_1(x/z) \rrbracket^{J,v}_{w,J(P,w)},\dots,\llbracket t_n(x/z) \rrbracket^{J,v}_{w,J(P,w)}} \in J(P,w),
            \]
            we can use $v(x) = v(y)$ and Proposition~\ref{prop:valuation} to obtain
            \[
                \tuple{\llbracket t_1(y/z) \rrbracket^{J,v}_{w,J(P,w)},\dots,\llbracket t_n(y/z) \rrbracket^{J,v}_{w,J(P,w)}} \in J(P,w).
            \]
            Thus $N,w \models_v P(t_1,\dots,t_n)(y/z)$.
            \item For $\varphi$ being of the forms $\neg \psi$ or $\psi \land \gamma$, the proof is straightforward.
            \item For $\varphi$ being of the form $\all z' \psi$, suppose $N,w \models_v x = y$ and $N,w \models_v (\all z' \psi)(x/z)$.
            If $z' = z$, obviously $N,w \models_v (\all z' \psi)(y/z)$. If $z' \neq z$, then we have $N,w \models_v \all z' \psi(x/z)$ thus $N,w \models_{v[z' \mapsto d]} \psi(x/z)$ for all $d \in D_{\type(z')}$. Since we have $N,w \models_{v[z' \mapsto d]} x = y$ for all $d \in D_{\type(z')}$, by inductive hypothesis we obtain $N,w \models_{v[z' \mapsto d]} \psi(y/z)$ for all $d \in D_{\type(z')}$. Therefore, $N,w \models_v (\all z' \psi)(y/z)$.
            \item For $\varphi$ being of the form $K_{t}\psi$, suppose $N,w \models_v x = y$ and $N,w \models_v (K_{t}\psi)(x/z)$. Then $N,w' \models_v \psi(x/z)$ for all $w' \in W$ such that $\tuple{w,w'} \in R_{\llbracket t(x/z) \rrbracket^{J,v}_{w,\emptyset}}$. Now, we have $N,w' \models_{v} x = y$ for all $w' \in W$, as well as $\llbracket t(x/z) \rrbracket^{J,v}_{w,\emptyset}$ $=$ $\llbracket t(y/z) \rrbracket^{J,v}_{w,\emptyset}$ by $v(x)$ $=$ $v(y)$ and Proposition~\ref{prop:valuation}. So by inductive hypothesis we obtain $N,w' \models_v \psi(y/z)$ for all $w' \in W$ such that $\tuple{w,w'} \in R_{\llbracket t(y/z) \rrbracket^{J,v}_{w,\emptyset}}$. Thus, $N,w \models_v (K_t \psi)(y/z)$.
        \end{itemize}

        \item For $\eid$, i.e., $c=c \to \some x (x = c)$, suppose $N,w \models_v c=c$. Since $N,w \models_{v[x \mapsto J(c,w,J(=,w))]} x = c$, we have $N,w \models_v \some x (x=c)$, as required.
        \item For $\dd$, i.e., $x \neq y$ if $\type(x) \neq \type(y)$, suppose $\type(x) \neq \type(y)$ and let $N$, $w$ and $v$ be arbitrary. By the definition of valuation, each of $v(x)$ and $v(y)$ is in $D_{\type(x)}$ and $D_{\type(y)}$, respectively. Since $\type(x) \neq \type(y)$, $D_{\type(x)}$ and $D_{\type(y)}$ must be disjoint. Thus $N,w \models_v x \neq y$, as required.
        \item For $\axk$, i.e., $K_{t}(\varphi \to \psi) \to (K_{t}\varphi \to K_{t}\psi)$, suppose $N,w \models_v K_{t}(\varphi \to \psi)$ and $N,w \models_v K_{t}\varphi$. Let $w'$ be any world such that $\tuple{w,w'} \in R_{\llbracket t \rrbracket^{J,v}_{w,\emptyset}}$. Then we have $N,w' \models_v \varphi \to \psi$ and $N,w' \models_v \varphi$. Thus $N,w' \models_v \psi$, as required.
        \item For $\barcan$, i.e., $\all x K_{t}\varphi \to K_{t} \all x \varphi$ for $x$ not occurring in $t$, suppose $N,w \models_v \all x K_{t}\varphi$. To show $N,w \models_{v} K_{t} \all x \varphi$, let $w'$ be any world such that $\tuple{w,w'} \in R_{\llbracket t \rrbracket^{J,v}_{w,\emptyset}}$ and take any $d \in D_{\type(x)}$. By our supposition, we have $N,w \models_{v[x \mapsto d]} K_{t} \varphi$. Now $\llbracket t \rrbracket^{J,v}_{w,\emptyset}$ $=$ $\llbracket t \rrbracket^{J,v[x \mapsto d]}_{w,\emptyset}$ holds since $x$ does not occur in $t$. Thus $N,w' \models_{v[x \mapsto d]} \varphi$, as required.
        \item For $\kni$, i.e., $x \neq y \to K_{t} x \neq y$, suppose $N,w \models_v x \neq y$. By definition, obviously $N,w' \models_v x \neq y$ for all worlds $w'$. Thus $N,w \models K_{t} x \neq y$, as required.
    \end{itemize}
    By the above argument the proof has completed.
\end{proof}

We can now prove the semantic incompleteness of $\Hi\K$ as follows.

\begin{theorem}
\label{thm:unprovability}
Let $\Sigma = \tuple{\VAR,\CON,\FUN,\REL,\type}$ be a signature, $x \in \VAR$, $c \in \CON$ with $\type(x) = \type(c)$ and $P \in \REL$ with $\type(P)$ $=$ $\tuple{\agtobj}$. A formula $x = c \to (P(x) \to P(c))$ is not provable in $\Hi\K$.
\end{theorem}

\begin{proof}
    If $x = c \to (P(x) \to P(c))$ is provable in $\Hi\K$, then by the soundness (Theorem~\ref{thm:soundness}) it must be valid in the non-standard semantics, which contradicts Proposition~\ref{prop:invalidity-of-unprovable-formula}.
\end{proof}

\begin{corollary}[Semantic Incompleteness of $\Hi\K$]
\label{cor:semantic-incompleteness}
    The Hilbert-style system $\Hi\K$ is semantically incomplete with respect to the TML-semantics, i.e., there exists some formula $\varphi$ such that $\varphi$ is valid in the TML-semantics but not provable in $\Hi\K$.
\end{corollary}

\begin{proof}
    By Proposition~\ref{prop:validity-of-unprovable-formula} and Theorem~\ref{thm:unprovability}.
\end{proof}

\section{Conclusion}

In this paper, we proved that Liberman et~al.\cite{Liberman2020}'s Hilbert-style system $\Hi\K$ for the term-modal logic $\K$ with equality and non-rigid terms is semantically incomplete by introducing the non-standard semantics for which $\Hi\K$ is sound but in which $x = c \to (P(x) \to P(c))$ is not valid.

A further direction to be pursued is to give sound and complete Hilbert-style systems for term-modal logics including $\K$ with equality and non-rigid terms. Such systems, for example, might be obtained as slight modifications of the system given in Fagin et~al.~\cite[p. 90]{Fagin2003}. Another further direction that might be worth studying is to apply the non-standard semantics to the analysis of natural language. As Example~\ref{ex:lewis} suggests, it is reasonable to see $J(P)$ in $J(c,w,J(P))$ as a kind of \textit{context} uniquely determining the denotation of a constant $c$ at a world $w$. Thus, the non-standard semantics might be seen as a semantics capturing the context-dependency of the denotations of nouns in natural language.

\section*{Acknowledgement}
This work was supported by JSPS KAKENHI Grant Number JP23KJ2146.

\bibliographystyle{eptcs}

\end{document}